\documentclass[a4paper,UKenglish,cleveref, autoref, thm-restate]{lipics-v2021}

\title{Parameterized Complexity of Propositional Inclusion and Independence Logic}

\titlerunning{Parameterized Complexity of PINC and PIND}

\author{Yasir Mahmood}{DICE group, Department of Computer Science, Paderborn University, Germany}{yasir.mahmood@uni-paderborn.de}{https://orcid.org/0000-0002-5651-5391}{}{}

\author{Jonni Virtema}{Department of Computer Science, University of Sheffield, United Kingdom}{j.t.virtema@sheffield.ac.uk}{https://orcid.org/0000-0002-1582-3718}{}{}

\authorrunning{Y.~Mahmood and J.~Virtema }

\Copyright{Yasir Mahmood and Jonni Virtema} 



\funding{This work was supported by the European Union’s Horizon Europe research and innovation programme within project ENEXA (101070305) and by the German Research Foundation (DFG), project VI 1045/1-1.}

\nolinenumbers 

\hideLIPIcs  

\EventEditors{John Q. Open and Joan R. Access}
\EventNoEds{2}
\EventLongTitle{42nd Conference on Very Important Topics (CVIT 2016)}
\EventShortTitle{CVIT 2016}
\EventAcronym{CVIT}
\EventYear{2016}
\EventDate{December 24--27, 2016}
\EventLocation{Little Whinging, United Kingdom}
\EventLogo{}
\SeriesVolume{42}
\ArticleNo{23}
\usepackage{macros}
\begin{document}
\maketitle

\begin{abstract}
 We give a comprehensive account on the parameterized complexity of model checking and satisfiability of propositional inclusion and independence logic. We discover that for most parameterizations the problems are either in FPT or paraNP-complete.

\end{abstract}

\section{Introduction}



The research program on team semantics was conceived in the early 2000s
to create a unified framework to study logical foundations of different notions of dependence between variables.
Soon after the introduction of first-order dependence logic \cite{vaananen07}, the framework was extended to cover propositional and modal logic \cite{vaananen08}. 
In this context, a significant step was taken in \cite{ebbing13}, where the focus shifted to study dependencies between formulas instead of variables. 
The framework of team semantics has been proven to be remarkably malleable. During the past decade the framework has been re-adapted for the needs of an array of disciplines.
In addition to the modal variant, team semantics has been generalized to temporal \cite{KrebsMV15} and probabilistic \cite{HKMV18} frameworks, and fascinating connections to fields such as database theory  \cite{HannulaKV20}, statistics \cite{CoranderHKPV16}, real valued computation \cite{HannulaKBV20}, verification \cite{kmvz18}, and quantum information theory \cite{Hyttinen15b} have been identified.

Boolean satisfiability problem (SAT) and quantified Boolean formula problem (QBF) have had a widespread influence in diverse research communities. 
In particular, QBF solving techniques are important in application domains such as planning, program synthesis and
verification, adversary games, and non-monotonic reasoning, to name a few \cite{DBLP:conf/ictai/ShuklaBPS19}. Further generalizations of QBF are the dependency quantified Boolean formula problem (DQBF) and alternating DQBF which allow richer forms of variable dependence \cite{HLKV16,PETERSON2001957,DBLP:conf/focs/PetersonR79}. Propositional logics with team semantics offer a fresh perspective to study enrichments of SAT and QBF. 
Indeed, the so-called \emph{propositional dependence logic} (\PDL) is known to coincide with DQBF,
whereas quantified propositional logics with team semantics have a close connection to alternating DQBF \cite{HLKV16,DBLP:journals/iandc/Virtema17}.

Propositional dependency logics extend propositional logic with atomic dependency statements describing various forms of variable dependence.
In this setting, formulas are evaluated over propositional teams (i.e, sets of propositional assignments with common variable domain). 
An \emph{inclusion atom} $\inca{x}{y}$ is true in a team $T$, if $\forall s \in T \,\exists t\in T$ such that $s(\tuple x)=t(\tuple y)$. 
An \emph{independence atom} $\indepa{x}{y}{z}$ expresses that in a team $T$, for any fixed value for the variables in $\tuple{z}$ the values for $\tuple{x}$ and $\tuple{y}$ are informationally independent.
The extension of propositional logic with 
inclusion and independence atoms yield
inclusion ($\pinc$) and independence ($\pind$) logics, respectively.
\begin{example}
	Table~\ref{database:example} illustrates an example from relational databases.
	\begin{table}[t]
		\centering
		\begin{tabular}{c@{\; }c@{\; }c@{\; }c@{\; }c@{\; }}\toprule
			\texttt{Instructor} 	& 	\texttt{Time}		& 	\texttt{Room} 	& 	\texttt{Course} & \texttt{Responsible} \\ \toprule
			Antti	& 09:00	& A.10	& Genetics	& Antti 		\\
			Antti	& 11:00	& A.10 	& Chemistry	& Juha  \\
			Antti	& 15:00	& B.20 	& Ecology	& Antti  \\
			Jonni	& 10:00	& C.30 	& Bio-LAB	& Jonni 	 \\
			Juha	& 10:00	& C.30 	& Bio-LAB	& Jonni 	 \\
			Juha	& 13:00	& A.10 	& Chemistry	& Juha \\
			\bottomrule
		\end{tabular} \quad
		\begin{tabular}{c@{\; }c@{\; }c@{\; }c@{\; }c@{\; }}\toprule
			$i_1i_2$ 	& $t_1t_2t_3$ & 	$r_1r_2$ 	& 	$c_1c_2$ & $p_1p_2$\\ \toprule
			00		&	110 	&	11	&	11 & 00	\\
			00		&	111		&	11	&	00 & 10 \\
			00		&	000		&	00	&	01 & 00\\
			01		&	001		&	01	&	10 & 01\\
			10		&	001		&	01	&	10 & 01	\\
			10		&	010		&	11	&	00 & 10	\\
			\bottomrule
		\end{tabular}	
		\caption{(Left) An example database with $5$ attributes and universe size $15$. (Right) An encoding with $3\cdot \lceil\log_2(3)\rceil+\lceil\log_2(5)\rceil+\lceil\log_2(4)\rceil$ many propositional variables.}\label{database:example}
	\end{table}	
	The set of records corresponds to a team, that satisfies
	the dependency $\incas{\texttt{Responsible}}{\texttt{Instructor}}$. Moreover, it
	violates the independence $\indepas{\texttt{Instructor}}{\texttt{Time}}{\texttt{Course}}$ as witnessed by tuples (Antti, 11:00, A.10, Chemistry, Juha) and (Juha, 13:00, A.10, Chemistry, Juha).
	%
	%
	In propositional logic setting, datavalues can be represented as bit strings of appropriate length (as depicted in Table~\ref{database:example}).
\end{example}

The complexity landscape of the classical (non-parameterized) decision problems --- satisfiability, validity, and model checking --- is well mapped  in the propositional and modal team semantics setting (see 
\cite[page 627]{HellaKMV19} for an overview). 
Parameterized complexity theory,
pioneered by Downey and Fellows \cite{DBLP:series/txcs/DowneyF13},
is a widely studied subarea of complexity theory. 
The motivation being that it provides a deeper analysis than the classical complexity theory by providing further insights into the source of intractability.
The idea here is to identify meaningful parameters of inputs such that fixing those makes the problem tractable. One example of a fruitful parameter is the treewidth of a graph. 
A parameterized problem (PP) is called fixed parameter tractable, or in $\FPT$ for short, if for a given input $x$ with parameter $k$, the membership of $x$ in PP can be decided in time $f(k)\cdot p(|x|)$ for some computable function $f$ and polynomial $p$.
That is, for each fixed value of $k$ the problem is tractable in the classical sense of tractability (in $\Ptime$), and the degree of the polynomial is independent of the parameter.
%
The class $\para\NP$ consists of problems decidable in time $f(k)\cdot p(|x|)$ on a non-deterministic machine.

In the propositional team semantics setting, the study of parameterized complexity was initiated by Meier and Reinbold \cite{MeierR18} in the context of parameterized enumeration problems, and by Mahmood and Meier \cite{jpdl} in the context of classical decision problems, for $\PDL$.
In the first-order team semantics setting, Kontinen~et~al. \cite{KontinenMM22} studied parameterized model checking of dependence and independence logic, and in \cite{kmmv23} introduced the \emph{weighted-definability} problem for dependence, inclusion and independence logic thereby establishing a connection with the parameterized complexity classes in the well-known \emph{\W{}-hierarchy}. 

We focus on the parameterized complexity of model checking ($\MC$) and satisfiability ($\SAT$) of propositional inclusion and independence logic.
We consider both \emph{lax} and \emph{strict} semantics.
The former is the prevailing semantics in the team semantics literature. The past rejection of strict semantics was based on the fact that it does not satisfy locality \cite{DBLP:journals/apal/Galliani12} (the locality principle dictates that satisfaction of a formula should be invariant on the truth values of variables that do not occur in the formula).
Recent works have revealed that locality of strict semantics can be recovered by moving to multiteam semantics (here teams are multisets) \cite{DurandHKMV18}. Since, in propositional team semantics, the shift from teams to multiteams has no complexity theoretic implications, we stick with the simpler set based semantics for our logics. 
\begin{table}[t]
	\centering
		\begin{tabular}{l@{\; }c@{\; }c@{\; }c@{\; }c@{\; }} 
			\toprule
			&	\multicolumn{2}{l}{\qquad $\pind$}		& \multicolumn{2}{l}{\qquad $\pinc$}			\\
			Parameter 	& $\MC$ 	  &  $\SAT$   & $\MC_s$ 	  &  $\SAT$    \\ \toprule
			\formulatw 	 
			& $\para\NP^{\ref{thm:ind-mc}}$
			& $\FPT ^{\ref{thm:ind-sat}}$   
			& $\para\NP ^{\ref{thm:mc-many}}$
			& in $\para\NP ^{\ref{sat:tw}}$	\\
			\formulateamtw 	 
			& $\FPT^{\ref{MC:cor}}$	
			& - 
			& $\FPT^{\ref{MC:cor}}$	
			& -		\\
			\teamsize  	 	
			& $\FPT^{\ref{MC:general}}$	
			& -
			& $\FPT^{\ref{MC:general}}$	
			& -				\\
			\formula 	 
			& $\FPT ^{\ref{MC:cor}}$	 
			&  Trivial	
			& $\FPT ^{\ref{MC:cor}}$	 
			&  Trivial		\\
			\formuladepth 	 	
			& $\FPT ^{\ref{MC:cor}}$	
			&  $\FPT^{\ref{SAT:general}}$ 			
			& $\FPT ^{\ref{MC:cor}}$	 
			&  $\FPT^{\ref{SAT:general}}$ 	\\
			\variables   	
			& $\FPT^{\ref{MC:cor}}$ 
			& $\FPT^{\ref{SAT:general}}$ 	
			& $\FPT ^{\ref{MC:cor}}$	 
			& $\FPT^{\ref{SAT:general}}$ 	\\
			\splits  	 	
			& $\para\NP ^{\ref{thm:ind-mc}}$
			& $\FPT ^{\ref{thm:ind-sat}}$ 
			& $\para\NP ^{\ref{thm:mc-many}}$
			& $\Ptime ^{\ref{thm:no-split}}$ if $\splits{=}0 $ 	\\
			\arity
			& $\para\NP ^{\ref{thm:ind-mc}}$
			& $\para\NP ^{\ref{thm:ind-sat}}$
			& $\para\NP ^{\ref{thm:mc-many}}$
			& $\para\NP ^{\ref{SAT:inc-arity}}$ \\
			\bottomrule
		\end{tabular}
	\caption{Overview of parameterized  complexity results with pointers to the results. The $\para\NP$-cases are complete, except for only membership in the first row. $\MC_s$ denotes model checking for strict semantics whereas $\MC/\SAT$ refer to both semantics. }\label{table:para}
\end{table}
In the model checking problem, one is given a team $T$ and a formula $\phi$, and the task is to determine whether $T\models \phi$.
In the satisfiability problem, one is given a formula $\phi$, and the task is to decide whether there exists a non-empty satisfying team $T$ for $\phi$.
Table~\ref{table:para} gives an overview of our results. We consider only strict semantics for $\MC$ of $\pinc$, since for lax semantics the problem is tractable already in the non-parameterized setting \cite[Theorem~$3.5$]{HellaKMV19}.\looseness=-1
\section{Preliminaries}
We assume familiarity with standard notions in complexity theory such as classes $\Ptime,\NP$ and $\EXP$ \cite{DBLP:books/daglib/0092426}.
We give a short exposition of relevant concepts from parameterized complexity theory.
For a broader introduction consider the textbook of Downey and Fellows \cite{DBLP:series/txcs/DowneyF13}, or that of Flum and Grohe \cite{DBLP:series/txtcs/FlumG06}. 

A \emph{parameterized problem} (PP) $\Pi\subseteq \Sigma^*\times\mathbb N$ consists of tuples $(x, k)$, where $x$ is called an instance and $k$ is the (value of the) parameter.

\vspace{2mm}
\noindent\textbf{$\FPT$ and $\para\NP$.}
Let $\Pi$ be a PP over $\Sigma^*\times\mathbb N$. 
Then $\Pi$ is \emph{fixed parameter tractable} ($\FPT$ for short) if it can be decided by a deterministic algorithm $\mathcal A$ in time $f(k)\cdot p(|x|)$ for any input $(x,k)$, where $f$ is a computable function and $p$ is a polynomial.
If the algorithm $\mathcal A$ is non-deterministic instead, then $\Pi$ belongs to the class $\para\NP$.	


%
%
The notion of hardness in parameterized setting is employed by fpt-reductions.

\vspace{2mm}
\noindent\textbf{fpt-reductions.}
Let $\Pi\subseteq\Sigma^*\times\mathbb N$ and $\Theta\subseteq\Gamma^*\times\mathbb N$ be two PPs.
Then $\Pi$ is \emph{fpt-reducible} to $\Theta$, if there exists an fpt-computable function  $g\colon\Sigma^*\times\mathbb N\to\Gamma^*\times\mathbb N$ such that (1) for all $(x,k)\in\Sigma^*\times\mathbb N$ we have that $(x,k)\in \Pi\Leftrightarrow g(x,k)\in \Theta$ and (2) there exists a computable function $h\colon\mathbb N\to\mathbb N$ such that for all $(x,k)\in\Sigma^*\times\mathbb N$ and $g(x,k)=(x',k')$ we have that $k'\leq h(k)$.


We will use the following result 
to prove \para\NP-hardness. Let $\Pi$ be a PP over $\Sigma^*\times\N$.
Then the \emph{$\ell$-slice of $\Pi$}, for $\ell\ge0$, is the set $\Pi_\ell\dfn\{\,x\mid (x,\ell)\in \Pi\,\}$. 
\begin{proposition}[{\cite[Theorem~2.14]{DBLP:series/txtcs/FlumG06}}]\label{slice-NP-result}
	Let $\Pi$ be a PP in $\para\NP$.
	If there exists an $\ell\ge0$ such that $\Pi_\ell$ is $\NP$-complete, then $\Pi$ is $\para\NP\complete$.
\end{proposition} 
\noindent Moreover, we will use the following folklore result to get several upper bounds.
\begin{proposition}\label{parameter-bound}
	Let $Q$ be a problem such that $(Q, k)$ is $\FPT$ and let $\ell$ be a parameter with $k\leq f(\ell)$ for some computable function $f$. Then $(Q, \ell)$ is $\FPT$.
\end{proposition}

\vspace{2mm}
\noindent\textbf{Propositional Team Based Logics.}
Let $\VAR$ be a countably infinite set of variables. 
The syntax of propositional logic (\pl) is defined via the following grammar:
$
\phi\ddfn
x \mid
\lnot x\mid
\phi\lor\phi\mid
\phi\land\phi,
$
where $x\in\VAR$. Observe that we allow only atomic negations. As usual $\top\dfn x\lor\neg x$ and $\bot\dfn x\land \neg x$.  
Propositional dependence logic $\pdl$ is obtained by extending $\pl$ by atomic formulas of the form $\depa{x}{y}$, where $\mathbf{x,y}\subset\VAR$ are finite tuples of variables.
Similarly, adding inclusion atoms $\inca{x}{y}$ where ($|\tuple x|= |\tuple y|$) and independence atoms $\indepa{x}{y}{z}$ gives rise to propositional inclusion ($\pinc$) and independence ($\pind$) logic, respectively.
When we wish to talk about any of the three considered logics, we simply write $\mathcal L$.
That is, unless otherwise stated, $\mathcal L\in \{\pdl,\pinc,\pind\}$.
For an assignment $s$ and a tuple $\tuple x= (x_1,\ldots,x_n)$, $s(\tuple x)$ denotes $(s(x_1),\ldots,s(x_n))$.

\vspace{2mm}
\noindent\textbf{Team Semantics.}
Let $\phi,\psi$ be $\mathcal L$-formulas and $\mathbf{x,y,z}\subset\VAR$ be finite tuples of variables.
A \emph{team} $T$ is a set of assignments $t\colon\VAR\to\setdefinition{0,1}$.
The satisfaction relation $\models$ is defined as follows:
	\begin{alignat*}{3}
		&T\models x && \text{ iff } \quad && \forall t\in T: t(x)=1, \\
		&T\models \lnot x && \text{ iff } && \forall t\in T: t(x)=0, \\
		&T\models \phi\land\psi && \text{ iff } && T\models\phi\text{ and }T\models\psi, \\
		&T\models \phi\lor\psi && \text{ iff } && \exists T_1, T_2 (T_1\cup T_2=T): T_1\models\phi\text{ and }T_2\models\psi, \\
		&T\models \inca{x}{y}\quad && \text{ iff } && \forall t\in T\,\exists t'\in T: t(\mathbf x)=t'(\mathbf y),\\
		&T\models \indepa{x}{y}{z}\quad && \text{ iff } && \forall t,t'\in T: t(\mathbf z)=t'(\mathbf z) \text{, } \exists t'': t''(\mathbf{xzy})=t(\mathbf{xz})t'(\mathbf y).
	\end{alignat*}
%
%
Intuitively, an inclusion atom $\inca{x}{y}$ is true if the value taken by $\tuple x$ under an assignment $t$ is also taken by $\tuple y$ under some assignment $t'$.
Moreover, the independence atom $\indepa{x}{y}{z}$ has the meaning that whenever the value for $\tuple z$ is fixed under two assignments $t$ and $t'$, then there is an assignment $t''$ which maps $\tuple {x}$ and $\tuple {y}$ according to $t$ and $t'$, respectively.
We can interpret the dependence atom $\depa{x}{y}$ as the independence atom $\indepa{y}{y}{x}$. The operator $\lor$ is also called a split-junction in the context of team semantics.
Note that in the literature there exist two semantics for the split-junction: \emph{lax} and \emph{strict} semantics (e.g., Hella et~al.~\cite{HellaKMV19}).
Strict semantics requires the ``splitting of the team'' to be a partition whereas lax semantics allows an ``overlapping'' of the team.
Regarding $\pdl$ and $\pind$, the complexity for $\SAT$ and $\MC$ is the same irrespective of the considered semantics. 
However, the picture is different for $\MC$ in $\pinc$ as depicted in \cite[page 627]{HellaKMV19}.
For any logic $\mathcal L$, we denote $\MC$ under strict (respectively, lax) semantics by $\MC_s (\MC_l)$. 
Moreover, $\MC_l$ is in $\Ptime$ for $\pinc$ and consequently, we have only $\MC_s$ in Table~\ref{table:para}.


\section{Graph Representation of the Input}
In order to consider specific structural parameters, we need to agree on a representation of an input instance.
We follow the conventions given in \cite{jpdl}. 
Well-formed $\mathcal L $-formulas, for every $\mathcal L \in \{\pdl,\pinc,\pind\}$, can be seen as binary trees (the syntax tree) with leaves as atomic subformulas (variables and dependency atoms).
Similarly to $\pdl$ \cite{jpdl}, we take the syntax structure (defined below) rather than syntax tree as a graph structure in order to consider treewidth as a parameter.
%
We use the same graph representation for each logic $\mathcal L$. 
That is, when an atom $\depa{x}{y}$ is replaced by either $\inca{x}{y}$ or $\indepa{x}{y}{\emptyset}$, the graph representation, and hence, the treewidth of this graph remains the same.
Also, in the case of $\MC$, we include assignments in the graph representation.
In the latter case, we consider the Gaifman graph of the structure that models both, the team and the input formula.

\vspace{2mm}
\noindent\textbf{Syntax Structure.}
Let 
$\phi$ be an $\mathcal L$-formula with propositions $\{x_1, \ldots, x_n\}$ and $T=\setdefinition{s_1, \ldots s_m}$ a team.
The \emph{syntax structure} $\strucA_{T,\phi}$ has the vocabulary, 
$
\tau_{\problemind} \dfn \{\, \var{}^1, \SF{}^1, \succcurlyeq^2, \DEP^2, \istrue{}^2,\isfalse{}^2, r, c_1, \ldots,c_m\,\},
$
where the superscript denote the arity of each relation.
The universe of $\strucA_\problemind$ is $\univA\dfn\SubForm{\phi}\cup \VAR(\phi) \cup \setdefinition{c_1^\strucA,\ldots, c_m^\strucA}$, where $\SubForm{\phi}$ and $\VAR(\phi)$ denote the set of subformulas and variables appearing in $\phi$, respectively. 
\begin{itemize}
	\item $\SF{}$ and $\var{}$ are unary relations: `is a subformula of $\phi$' and `is a variable in $\phi$'.\looseness=-1
	\item $\succcurlyeq$ is a binary relation such that $\psi \succcurlyeq^\strucA \alpha$ iff $\alpha$ is an immediate subformula of $\psi$. That is, either $\psi = \neg\alpha$ or there is a $\beta\in \SubForm{\phi}$ such that $\psi = \alpha \oplus \beta$ where $\oplus \in \{\land, \lor\}$. Moreover, $r$ is a constant symbol representing $\phi$.
	\item $\DEP$ is a binary relation connecting $\mathcal L$-atoms and its parameters. For example, if $\alpha=\inca{x}{y}$ and $x,y\in \tuple x \cup \tuple y$, then $\DEP(\alpha,x)$ and $\DEP( x, y)$ are true.
	\item The set $\setdefinition{c_1,\ldots,c_m}$ encodes the team $T$. Each $c_i \in \tau_\problemind$ corresponds to an assignment $s_i \in T$ for $i \leq m$, interpreted as $ c_i^{\strucA} \in {A}$. 
	\item $\istrue{}$ and $\isfalse{}$ relate $\var{}$ with $c_1,\ldots, c_m$. $\istrue{c, x}$ (resp., $\isfalse{c, x}$) is true iff $x$ is mapped to $1$ (resp., $0$) by the assignment in $T$ interpreted by $c$. 
\end{itemize}
The \emph{syntax structure} $\strucA_{\phi}$ over a vocabulary $\tau_\phi$ is defined analogously.
Here $\tau_\phi$ neither contains the team related relations nor the constants $c_i^\strucA$ for $1\leq i\leq m$.


\vspace{2mm}
\noindent\textbf{Gaifman graph.} Let $T$ be a team, $\phi$ an $\mathcal L$-formula, $\strucA_\problemind$ and $A$ as above. The \emph{Gaifman graph} $G_\problemind=(\univA,E)$ of the $\tau_\problemind$-structure $\strucA_\problemind$ is defined as 
$
E\dfn\big\{\,\{u,v\}\;\big|\; u,v\in\univA, \text{ such that there is an } R\in\tau_\problemind \text{ with }(u,v)\in R\,\big\}.
$
Analogously, we let $G_\phi$ to be the \emph{Gaifman graph} for the $\tau_\phi$-structure $\strucA_\phi$.

Note that for $G_\phi$ we have $E = \DEP\; \cup \succcurlyeq$ and for $G_\problemind$ we have that $E=\DEP\,\cup \succcurlyeq \cup\, \mathsf{isTrue}\cup\mathsf{isFalse}$.
%


\vspace{2mm}
\noindent\textbf{Treewidth.}
A \emph{tree decomposition} of a graph $G=(V,E)$ is a tree $T=(B,E_T)$, where the vertex set $B\subseteq\mathcal P(V)$ is a collection of \emph{bags} and $E_T$ is the edge relation such that (1) $\bigcup_{b\in B}b=V$, (2) for every $\setdefinition{u,v}\in E$ there is a bag $b\in B$ with $u,v\in b$, and (3) for all $v\in V$ the restriction of $T$ to $v$ (the subset with all bags containing $v$) is connected.
The \emph{width} of a tree decomposition $T=(B,E_T)$ is the size of the largest bag minus one: $\max_{b\in B}|b|-1$.
The \emph{treewidth} of a graph $G$ is the minimum over widths of all tree decompositions of $G$.
%
The treewidth of a tree is one.
Intuitively, it measures the tree-likeness of a given graph.
%
%

\begin{example}[Adapted from~\cite{jpdl}]\label{ex:gaifman-examples}
	Figure~\ref{fig:example-tw} represents the Gaifman graph of the syntax structure $\mathcal A_\phi$ (in middle) with a tree decomposition (on the right).
	Since the largest bag is of size $3$, the treewidth of the given decomposition is $2$.
	Figure~\ref{fig:example-twe} presents the Gaifman graph of $\mathcal A_\problemind$, that is, when the team $T=\{s_1, s_2\}=\{0011,1110\}$ is part of the input  
	(an assignment $s$ is denoted as $s(x_1\ldots x_4)$).
	%
	\begin{figure}[t]
		\centering
		\begin{tikzpicture}[scale=0.6, 
			level 1/.style={sibling distance=7.5em, level distance= 3 em},
			level 2/.style={sibling distance=4em, level distance= 4 em}, 
			level 3/.style={sibling distance=3.5em, level distance= 4 em}, 
			every node/.style={scale=0.7},
			edge from parent/.style={thin,-,black, draw},
			leaf/.style = {draw, circle}]
			\node  {\textcolor{black}{$\wedge_r$}}
			child {node {$\lor_2$}
				child {node {$x_3$}}
				child {node {$\neg$}
					child {node {$x_1$}}}}	
			child {node {$\lor_1$}	
				child{node {$\incas{x_3}{x_4}$ }}
				child{node {$\land_1$}
					child {node {$x_1 $}}
					child {node {$x_2$}}}
			};
		\end{tikzpicture} 
		\quad
		\begin{tikzpicture}[gate/.style={inner sep=1mm,draw,rounded corners,rectangle},scale=.55, every node/.style={scale=0.7}]
			\node (x1) at (0,0) {$x_1$};
			\node (x2) at (1,0) {$x_2$};
			\node (x3) at (1.8,0) {$x_3$};
			\node (x4) at (3.2,0) {$x_4$};
			
			\node[gate] (and1) at (.5,-1) {$\land_1$};
			\node[gate] (dep) at (2.5,-1) {$\mathsf{Inc}$};
			
			\node[gate] (not) at (0,-2) {$\lnot$};
			\node[gate] (or1) at (2.5,-2) {$\lor_1$};
			
			\node[gate] (or2) at (.5,-3) {$\lor_2$};
			
			\node[gate] (andr) at (1.5,-4) {$\land_r$};
			
			\foreach \f/\t in {x1/and1,x2/and1,x3/dep,x4/dep,and1/or1,dep/or1,x1/not,not/or2,x3/or2,or2/andr,or1/andr, x3/x4}{
				\draw[-] (\f) -- (\t);
			}
		\end{tikzpicture}
		\quad
		\begin{tikzpicture}[scale=0.6, level distance= 2.5 em, sibling distance= 6 em,
			every node/.style={draw, scale=0.6},
			edge from parent/.style={thin,-,black, draw},
			leaf/.style = {draw, circle}]
			\node  {\textcolor{black}{$\wedge_r,\lor_1,\lor_2$}}
			child {node {$\lor_1,\lor_2,\mathsf{Inc}$} 
				child {node {$\lor_2,\mathsf{Inc},x_3$}
					child {node {$\mathsf{Inc},x_3,x_4$}
				}}		
				child {node {$\lor_1,\lor_2,\land_1$}	
					child{node {$\lor_2,\land_1,x_2$}
						child{node {$\lor_2,\land_1,x_1 $}
							child{node {$\lor_2,x_1,\neg $}}}}}};
		\end{tikzpicture}
		
		\caption{An example syntax tree (left) with the corresponding Gaifman graph (middle) and a tree decomposition (right) for $(x_3\lor \neg x_1) \land \big(\incas{x_3}{x_4} \lor (x_1\land x_2)\big)$. We abbreviated subformulas in the inner vertices of the Gaifman graph for better presentation.}\label{fig:example-tw} 
	\end{figure}
\end{example}

\begin{figure}[t]
	\centering
	\begin{tikzpicture}[gate/.style={inner sep=1mm,draw,rounded corners,rectangle},scale=.5, every node/.style={draw,scale=0.65}]
		\node (x1) at (-0.5,0) {$x_1$};
		\node (x2) at (0.75,0) {$x_2$};
		\node (x3) at (2.1,0) {$x_3$};
		\node (x4) at (3.5,0) {$x_4$};
		
		\node (s1) at (-0.5,2) {$c_1$};
		\node (s2) at (3.5,2) {$c_2$};
		
		\node[gate] (and1) at (.5,-1) {$\land_1$};
		\node[gate] (dep) at (3,-1) {$\mathsf{Inc}$};
		
		\node[gate] (not) at (0,-2) {$\lnot$};
		\node[gate] (or1) at (2.7,-2) {$\lor_1$};
		
		\node[gate] (or2) at (.5,-3) {$\lor_2$};
		
		\node[gate] (andr) at (1.5,-4) {$\land_r$};
		
		\foreach \f/\t in {x1/and1,x2/and1,x3/dep,x4/dep,and1/or1,dep/or1,x1/not,not/or2,x3/or2,or2/andr,or1/andr}{
			\draw[-] (\f) -- (\t);
		}
		\foreach \f/\t in {x1/s1,x1/s2,x2/s1,x2/s2,x3/s1,x3/s2,x4/s1,x4/s2, x3/x4}{
			\draw[-] (\f) -- (\t);
		}
	\end{tikzpicture}
	\qquad
	\begin{tikzpicture}[scale=0.65, level distance= 2.5 em, sibling distance= 9 em,
		every node/.style={draw, scale=0.65},
		edge from parent/.style={thin,-,black, draw},
		leaf/.style = {draw, circle}]
		\node  {\textcolor{black}{$\wedge_r,\lor_1,\lor_2$}}
		child {node {$\lor_1,\lor_2,\mathsf{Inc}, c_1,c_2$} 
			child {node {$\lor_2,\mathsf{Inc},x_3, c_1,c_2$}
				child {node {$\mathsf{Inc},x_3, x_4,c_1,c_2$}
			}}		
			child {node {$\lor_1,\lor_2,\land_1, c_1,c_2$}	
				child{node {$\lor_2,\land_1,x_2, c_1,c_2$}
					child{node {$\lor_2,\land_1,x_1, c_1,c_2$}
						child{node {$\lor_2,x_1,\neg, c_1,c_2$}}}}}};
	\end{tikzpicture}
	
	\caption{The Gaifman graph for $\langle T,\Phi \rangle$ (Example~\ref{ex:gaifman-examples}) with a possible tree decomposition.}\label{fig:example-twe}
\end{figure}

\vspace{2mm}
\noindent\textbf{Parameterizations.}
We consider eight different parameters for $\MC$ and six for $\SAT$. 
For $\MC$, we include $\formulatw$, $\formulateamtw$, $\teamsize$, $\formula$, $\variables$, $\formuladepth$, $\splits$ and $\arity$.
However, for $\SAT$, $\formulateamtw$ and $\teamsize$ are not meaningful.
All these parameters arise naturally in problems we study. 
Let $T$ be a team and $\phi$ an $\mathcal L$-formula.
$\splits$ denotes the number of times a split-junction $(\lor)$ appears in $\phi$ and $\variables$ denotes the number of distinct propositional variables.
$\formuladepth$ is the depth of the syntax tree of $\phi$, that is, the length of the longest path from root to any leaf in the syntax tree. 
$\teamsize$ is the cardinality of the team $T$, and $\formula$ is $|\phi|$.
For a dependence atom $\depa{x}{y}$ and inclusion atom $\inca{x}{y}$, the arity is defined as $|\tuple x|$ (recall that $|\tuple x|= |\tuple y|$ for an inclusion atom), whereas, for an independence atom $\indepa{x}{y}{z}$, it is
the number of distinct variables appearing in $\indepa{x}{y}{z}$.  Finally, $\arity$ denotes the maximum arity of any $\mathcal L$-atom in $\phi$.
%
Regarding treewidth, recall that for $\MC$, we also include the assignment-variable relation in the graph representation.
This yields two graphs: $G_\phi$ for $\phi$, and $G_{T,\phi}$ for $\langle T,\phi\rangle$.
Consequently, there are two treewidth notions. $\formulatw$ is the treewidth of $G_\phi$ and $\formulateamtw$ is the treewidth of $G_{T,\phi}$. 
The name emphasises whether the team is also part of the graph.
As we pointed out $\formulateamtw$ and $\teamsize$ are both only relevant for $\MC$ because an instance of $\SAT$ does not contain a team.

Given an instance $\langle\problemind\rangle$ and a parameterisation $\kappa$, then $\kappa(\problemind)$ denotes the parameter value of $\langle\problemind\rangle$.
The following relationship between several of the aforementioned parameters was proven for $\pdl$.
It is easy to observe that the lemma also applies to $\pinc$ and $\pind$.
\begin{lemma}[\cite{jpdl}]\label{para:relations}
	Let $\mathcal L \in \{\pdl,\pinc,\pind\}$, $\phi$ an $\mathcal L$-formula and $T$ be a team. Then, $\teamsize(\problemind) \leq 2^{\variables(\problemind)}$, $\teamsize(\problemind ) \leq 2^{\formula(\problemind)}$, and $\formula(\problemind) \leq 2^{2\cdot \formuladepth(\problemind)}$.
\end{lemma}
Moreover, recall that we use the same graph representation for $\pdl, \pinc$ and $\pind$.
As a consequence, the following result also applies.
\begin{corollary}[\cite{jpdl}]\label{cor:tw-teamsize}
	Let $\mathcal L \in \{\pdl,\pinc,\pind\}$, $\phi$ an $\mathcal L$-formula and $T$ be a team. 
	Then $\formulateamtw(\problemind)$ bounds $\teamsize(\problemind)$. 
\end{corollary} 
\section{Complexity of inclusion and independence logic}
We start with general complexity results that hold for any team based logic whose atoms are $\Ptime$-checkable.
An atom $\alpha$ is $\Ptime$-checkable if given a team $T$, $T\models \alpha$ can be checked in polynomial time.
It is immediate that each atom considered in this paper is $\Ptime$-checkable. 
\begin{theorem}\label{MC:general}
	Let $\mathcal L$ be a team based logic such that $\mathcal L$-atoms are $\Ptime$-checkable, then $\MC$ for $\mathcal L$ when parameterized by $\teamsize$ is $\FPT$. 
\end{theorem}

\begin{proof}
	We claim that the bottom up (brute force) algorithm for the model checking of $\pdl$ \cite[Thm.~17]{jpdl} works for any team based logic $\mathcal L$ such that $\mathcal L$-atoms are $\Ptime$-checkable.
	The algorithm begins by checking the satisfaction of atoms against each subteam.
	This can be achieved in $\FPT$-time since teamsize and consequently the number of subteams is bounded.
	Moreover, by taking the union of subteams for split-junction and keeping the same team for conjunction the algorithm can find subteams for each subformula in $\FPT$-time.
	Lastly, it checks that the team $T$ is indeed a satisfying team for the formula $\phi$.
	For any team based logic $\mathcal L$, the $\FPT$ runtime is guaranteed if $\mathcal L$-atoms are $\Ptime$-checkable.
	Finally, the proof works for both strict and lax semantics.
\end{proof}

The following corollary to Theorem \ref{MC:general} is derived using Lemma~\ref{para:relations} and Proposition~\ref{parameter-bound}.

\begin{corollary}\label{MC:cor}
	Let $\mathcal L$ be a team based logic such that $\mathcal L$-atoms are $\Ptime$-checkable, then $\MC$ for $\mathcal L$ when parameterized by $k$ is $\FPT$, if $k\in\{\formulateamtw,\linebreak \formuladepth, \variables,\formula\}$. 
\end{corollary}
The following theorem states results for satisfiability.
\begin{theorem}\label{SAT:general}
	Let $\mathcal L$ be a team based logic s.t. $\mathcal L$-atoms are $\Ptime$-checkable, then $\SAT$ for $\mathcal L$ when parameterized by $k$ is $\FPT$, if $k\in \{\formuladepth, \variables\}$.
\end{theorem}

\begin{proof}
	Notice first that the case for $\formula$ is trivial because any problem parameterized by input size is $\FPT$.
	Moving on, bounding $\formuladepth$ also bounds $\formula$, this yields $\FPT$-membership for $\formuladepth$ in conjunction with Prop.~\ref{parameter-bound}.
	Finally, for $\variables$, one can enumerate all of the $2^{2^{\variables}}$-many teams in $\FPT$-time and determine whether any of these satisfies the input formula.
	The last step requires that the model checking parameterized by $\teamsize$ is $\FPT$, which is true due to Theorem~\ref{MC:general}.
	This completes the proof.
\end{proof}

Our main technical contributions are the following two theorems which establish that the satisfiability problem of $\pinc$ parameterized by $\arity$ is $\para\NP$-complete, and that $\SAT$ of $\pinc$ without disjunctions is tractable.
%
%
We start with the former.
The hardness follows from the $\NP$-completeness of $\pl$.
For membership, we give a non-deterministic algorithm $\mathbb A$ solving $\SAT$.
%
%
	%
%
%

\begin{theorem}\label{SAT:inc-arity}
	There is a non-deterministic algorithm $\mathbb A$ that, given a $\pinc$-formula $\phi$ with arity $k$, runs in $O(2^k\cdot p(|\phi|))$-time and outputs a non-empty team $T$ such that $T\models \phi$ if and only if $\phi$ is satisfiable.
\end{theorem}

\begin{proof}
	We present the proof for lax semantics first, towards the end we describe some modifications that solve the case for strict semantics.
	Given an input $\pinc$-formula $\phi$, the algorithm $\mathbb A$ operates on the syntax tree of $\phi$ and constructs a sequence of teams $\recursiveteam{i}{\psi}$ for each $\psi\in \SubForm{\phi}$ as follows.
	We let $\recursiveteam{0}{\psi}\dfn\emptyset$ for each $\psi\in\SubForm{\phi}$.
	Then, $\mathbb A$ begins by non-deterministically selecting a singleton team 
	$\recursiveteam{1}{\phi}$ for $\phi$.
	For $i\geq 1$, $\mathbb A$ implements the following steps recursively.
	
	\begin{enumerate} 
		\item[]	For odd $i\in \mathbb N$, $f_i(\psi)$ is defined in a \emph{top-down} fashion as follows.
		\item $\recursiveteam{i}{\phi}\dfn \recursiveteam{i-1}{\phi}$ for $i\geq 3$.
		\item If $\psi = \psi_0\land\psi_1$, let $\recursiveteam{i}{\psi_0}\dfn\recursiveteam{i}{\psi}$ and $\recursiveteam{i}{\psi_1}\dfn\recursiveteam{i}{\psi}$. \label{algo:item-c}
		\item If $\psi = \psi_0\lor\psi_1$, then non-deterministically select two teams $P_0,P_1$ such that  $P_0\cup P_1 = \recursiveteam{i}{\psi}\setminus \recursiveteam{i-1}{\psi}$ and set $\recursiveteam{i}{\psi_j} \dfn \recursiveteam{i-1}{\psi_j}\cup P_j$ for $j=0,1$.\label{algo:item-d}
		\vspace{2mm}
		\item[] For even $i$, $\recursiveteam{i}{\psi}$ is defined in a \emph{bottom-up} fashion as follows.
		\item 
		If $\psi \in \SubForm\phi$ is an atomic literal, then immediately reject if $\recursiveteam{i-1}{\psi}\not\models \psi$ and set $\recursiveteam{i}{\psi}\dfn\recursiveteam{i-1}{\psi}$ otherwise.
		If $\psi \in \SubForm\phi$ is an inclusion atom, then construct $\recursiveteam{i}{\psi}\supseteq  \recursiveteam{i-1}{\psi}$ such that $\recursiveteam{i}{\psi}\models \psi$. For $\psi\dfn \inca{x}{y}$, this is done by (I) adding partial assignments $t(\tuple y)\dfn s(\tuple x)$ whenever an assignment $s$ is a cause for the \emph{failure} of $\psi$, and (II) non-deterministically selecting extensions of these assignments to the other variables.  \label{algo:item-lit}   
		
		\item If $\psi = \psi_0\land\psi_1$, or $\psi = \psi_0\lor\psi_1$ let $\recursiveteam{i}{\psi}\dfn\recursiveteam{i}{\psi_0}\cup \recursiveteam{i}{\psi_1}$. \label{algo:item-cd}
		
	\end{enumerate}  
	$\mathbb A$ terminates by accepting when a fixed point is reached. That is, we obtain $j\in\mathbb N$ such that $\recursiveteam{i}{\psi}=\recursiveteam{i+1}{\psi}$ for each $\psi\in \SubForm{\phi}$ when $i\geq j$. Moreover, $\mathbb A$ rejects if Step~\ref{algo:item-lit} triggers a rejection. 
	Notice that the only step when new assignments are added is at the atomic level.
	Whereas the split in Step~\ref{algo:item-d} concerns those assignments which arise from other subformulas through union in Step~\ref{algo:item-cd}. 
	We first prove the following claim regarding 
	the overall runtime for $\mathbb A$.

	\noindent\textbf{Claim I. }$\mathbb A$ runs in at most $O(2^k\cdot p(|\phi|))$ steps for some polynomial $p$, where $k$ is the arity of $\phi$. That is, a fixed point or rejection is reached in this time.
	
	\noindent\textbf{Proof of Claim. }
	In each iteration $i$, either $\mathbb A$ rejects, or keeps adding new assignments.
	Furthermore, new assignments are added only in the cases for inclusion atoms. 
	As a result, if $\mathbb A$ has not yet reached a fixed point the reason is that some inclusion atom has generated new assignments.
	Since we take union of subteams in the bottom-up step, the following top-down iteration in Steps $2$ and $3$ may also add assignments in a subteam.
	That is, the subteams from each $\psi\in\SubForm\phi$ are propagated to other subformulas during each iteration.
	Now, each inclusion atom of arity $l\leq k$ can generate at most $2^l$ new assignments due to Step $4$ in the algorithm.
	Let $n$ denote the number of inclusion atoms in $\phi$ and $k$ be their maximum arity. 
	Then $\mathbb A$ iterates at most $2^k\cdot c  n$ times, where $c$ is some constant due to the propagation of teams to other subformulas.
	%
	This implies that, if no rejection has occured,  there is some $j\leq 2^k\cdot cn$ such that $\recursiveteam{i}{\psi}=\recursiveteam{j}{\psi}$ for each subformula $\psi\in\SubForm{\phi}$ and $i\geq j$. 
	We denote this fixed point by $\recursiveteam{\infty}{\psi}$ for each $\psi\in\SubForm{\phi}$.
	
	Now, we analyze the time it takes to compute each iteration.
	For odd $i\geq 1$, Steps $1$ and $2$ set the same team for each subformula and therefore take linear time.
	Notice that the size of team in each iteration is bounded by $2^k\cdot n$.
	This holds because new assignments are added only in the case of inclusion atoms and $\mathbb A$ starts with a singleton team.
	Consequently, Step $3$ non-deterministically splits the teams of size $2^k\cdot n$ in each iteration $i$ for odd $i\geq 1$.
	Moreover, Step $4$ for even $i$ requires (1) polynomial time in $|\recursiveteam{i}{\psi}|$, if $\psi$ is an atomic literal, and (2) non-deterministic polynomial time in $2^l\cdot |\recursiveteam{i}{\psi}|$ if $\psi$ is an inclusion atom of arity $l\leq k$.
	Finally, the union in Step $5$ again requires linear time.
	This implies that each iteration takes at most a runtime of $2^k\cdot p(|\phi|)$ for some polynomial $p$.
	This completes the proof of Claim 1.

	We now prove that $\mathbb A$ accepts the input formula $\phi$ if and only if $\phi$ is satisfiable.	 
	%
	Suppose that $\mathbb A$ accepts and let $\recursiveteam{\infty}{\phi}$ denote the fixed point.
	We first prove by induction that $\recursiveteam{\infty}{\psi}\models \psi$ for each subformula $\psi$ of $\phi$.
	Notice that there is some $i$ such that $\recursiveteam{\infty}{\psi}= \recursiveteam{i}{\psi}$.
		The case for atomic subformulas is clear due to the Step~4 of $\mathbb A$.
		For conjunction, observe that the team remains the same for each conjunct. That is, when $\psi=\psi_0\land\psi_1$ and the claim holds for $\recursiveteam{\infty}{\psi_i}$ and $\psi_i$, then $\recursiveteam{\infty}{\psi}\models \psi_0\land \psi_1$ is true.
		For disjunction, if $\psi=\psi_0\lor\psi_1$ and $\recursiveteam{\infty}{\psi_i}$ are such that $\recursiveteam{\infty}{\psi_i}\models \psi_i$ for $i=0,1$, 
		then we have that $\recursiveteam{\infty}{\psi}\models \psi$ where $\recursiveteam{\infty}{\psi} = \recursiveteam{\infty}{\psi_0}\cup\recursiveteam{\infty}{\psi_1}$.
	In particular $\recursiveteam{\infty}{\phi} \models \phi$ and the correctness of our algorithm follows.
	
	For the other direction,
	suppose $\phi$ is satisfiable and $T$ is a witnessing team.
	Then there exists a labelling function for $T$ and $\phi$, given as follows. 
	\begin{enumerate}[I.]
		\item The label for $\phi$ is $T$. 
		\item For every subformula $\psi= \psi_0\oplus \psi_1 $ with subteam label $P\subseteq T$, the subteam label for $\psi_i$ is $P_i$ ($i=0,1$) such that we have $P_0=P_1=P$, if $\oplus=\land$, and $P_0\cup P_1=P$ if $\oplus=\lor$,	
		\item $P_\psi\models \psi$ for every $\psi\in\SubForm{\phi}$ with label $P_\psi$.
	\end{enumerate}
	%
	Then we prove that there exists an accepting path when the non-deterministic algorithm $\mathbb A$ operates on $\phi$.
	We claim that when initiated on a subteam $\{s\}\subseteq T$, $\mathbb A$ constructs a fixed point $\recursiveteam{\infty}{\phi} $ and halts by accepting $ \phi$.
	Recall that $\mathbb A$ propagates teams back and forth until a fixed point is reached.
	Moreover, the new assignments are added only at the atomic level. 
	Let $\alpha\dfn \inclusion{x}{y}$ be an inclusion atom such that $\recursiveteam{i}{\alpha}\not=\emptyset$ for odd $i$, then $\mathbb A$ constructs a subteam $\recursiveteam{i+1}{\alpha} \supseteq \recursiveteam{i}{\alpha}$  
	(on a non-deterministic branch) containing assignments $t$ from $P_\alpha$ such that $\recursiveteam{i+1}{\alpha}\models\alpha$.
	Since, there are at most $2^{|\tuple y|}$-many different assignments for $\tuple y$, we know that Step $4$ applies to $\alpha$ at most $2^{|\tuple y|}$ times.
	That is, once all the different assignments for $\tuple y$ have been checked in some iteration $i$: Step $4$ does not add any further assignments to $\recursiveteam{i'}{\alpha}$ for $i'\geq i+1$.
	%
	Finally, since there is a non-empty team $T$ such that $T\models \phi$, this implies that $\mathbb A$ does not reject $\phi$ in any iteration (because there is a choice for $\mathbb A$ to consider subteams guaranteed by the labelling function).
	Consequently, $\mathbb A$ accepts by constructing a fixed point in at most $O(2^k\cdot p(|\phi|))$-steps (follows from Claim I). 
	This completes the proof and establishes the correctness.
	
	A minor variation in the algorithm $\mathbb A$ solves $\SAT$ for the strict semantics.
	When moving downwards, $\mathbb A$ needs to ensure that an assignment goes to only one side of the split.
	Moreover, since the subteams are selected non-deterministically for atomic subformulas, (in the bottom-up iteration) only subteams which can split according to the strict semantics are considered.
\end{proof}






%
\begin{example}
	We include an example to explain how $\mathbb A$ from the proof of Theorem~\ref{SAT:inc-arity} operates.
	Figure~\ref{fig:algo} depicts the steps of $\mathbb A$ on a formula $\phi$. 
	An assignment over $\{x_1,\dots,x_4\}$ is seen as a tuple of length four.
	It is easy to observe that the third iteration already yields a fixed point and that $\recursiveteam{3}{\psi}=\recursiveteam{4}{\psi}$ for each $\psi\in\SubForm{\phi}$.
	In this example, the initial guess made by $\mathbb A$ is the team $\{0110\}$.
	\begin{figure}[t]
		\centering
		\begin{tikzpicture}[scale=0.7, 
			level 1/.style={sibling distance=7.5em, level distance= 3 em},
			level 2/.style={sibling distance=4em, level distance= 4 em}, 
			level 3/.style={sibling distance=3.5em, level distance= 4 em}, 
			every node/.style={scale=0.65},
			edge from parent/.style={thin,-,black, draw},
			leaf/.style = {draw, circle}]
			\node  {\textcolor{black}{$\wedge_r$}}
			child {node {$\lor_2$}
				child {node {$x_3$}}
				child {node {$\neg x_1$}}}	
			child {node {$\lor_1$}	
				child{node {$\incas{x_3}{x_4}$ }}
				child{node {$\land_1$}
					child {node {$x_1 $}}
					child {node {$x_2$}}}
			};
		\end{tikzpicture}
		\quad
				%
				%
		%
		\begin{tikzpicture}[every node/.style={scale=0.75}]
			
			\node  {\begin{tabular}{c@{\; }c@{\; }c@{\; }c@{\; }c@{\; }c@{\; }c@{\; }c}\toprule
					$\psi$ 					& $\recursiveteam{1}{\psi}$ 			& $\recursiveteam{2}{\psi}$ & $\recursiveteam{3}{\psi}$ \\ \toprule
					{$\phi$}		&		0110 & 0110,1111  & 0110,1111	\\
					{$\lor_2$}		&		0110 & 0110  & 0110,1111	\\
					{$\lor_1$}		&		0110 & 0110,1111  & 0110,1111	\\
					{$\incas{x_3}{x_4}$}&	0110 & 0110,1111  & 0110,1111	\\
					{$x_1\land x_2$}&		 	 &   & 1111	\\
					{$x_3$} 		&		0110 & 0110  & 0110,1111	\\
					{$\neg x_1$}	&		0110 & 0110  & 0110 	\\
					
					\bottomrule
				\end{tabular}				
			};
			
					%
		\end{tikzpicture} 
		\caption{The table (Right) indicates subteams for each $\psi\in \SubForm{\phi}$ (Left). The teams $\recursiveteam{1}{\psi}$ and $\recursiveteam{3}{\psi}$ are propagated top-down whereas $\recursiveteam{2}{\psi}$ is propagated bottom-up. For brevity we omit subformulas $x_1$ and $x_2$ of $x_1\land x_2$.}\label{fig:algo}
	\end{figure}
	
\end{example}
The following corollaries follow immediately from the proof of Theorem~\ref{SAT:inc-arity}.

\begin{corollary}\label{inc-arity-lemma}
	Given a $\pinc$-formula $\phi$ with arity $k$, then 
	$\phi$ is satisfiable if and only if there is a team $T$ of size at most $O(2^k\cdot p(|\phi|))$ such that $T\models \phi$. 
\end{corollary}

\begin{proof}
	Simulate the algorithm $\mathbb A$ from the proof of Theorem~\ref{SAT:inc-arity}. 
	Since $\phi$ is satisfiable, $\mathbb A$ halts in at most $O(2^k\cdot p(|\phi|))$-steps and thereby yields a team (namely, $\recursiveteam{\infty}{\phi}$) of the given size.
\end{proof}
%


\begin{corollary}\label{sat:tw}
	$\SAT$ for $\pinc$, when parameterized by $\formulatw$ of the input formula is in $\para\NP$. 
\end{corollary}
\begin{proof}
	Recall the Graph structure where we allow edges between variables within an inclusion atom.
	This implies that for each inclusion atom $\alpha$, there is a bag in the tree decomposition that contains all variables of $\alpha$.
	As a consequence, a formula $\phi$ with treewidth $k$ has inclusion atoms of arity at most $k$.
	Consequently, $\SAT$ parameterized by treewidth of the input formula can be solved using the $\para\NP$-time algorithm from the proof of Theorem~\ref{SAT:inc-arity}.
\end{proof}
Regarding the parameter $\splits$, the precise parameterized complexity is still open for now.
However, we prove that if there is no split in the formula, then $\SAT$ can be solved in polynomial time. 
This case is interesting in its own right because it gives rise to the so-called Poor Man's $\pinc$, similar to the case of Poor Man's $\pdl$ \cite{DBLP:conf/sofsem/EbbingL12,DBLP:journals/sLogica/LohmannV13,MeierR18}. 
The model checking for this fragment is in $\Ptime$; this follows from the fact that $\MC$ for $\pinc$ with lax semantics is in $\Ptime$.
In the following, we prove that $\SAT$ for Poor Man's $\pinc$ is also in $\Ptime$.
%
\begin{theorem}\label{thm:no-split}
	There is a deterministic algorithm $\mathbb B$ that given a $\pinc$-formula $\phi$ with no splits runs in $\Ptime$-time and accepts if and only if $\phi$ is satisfiable.
\end{theorem}
%

\begin{proof}
	We give a recursive labelling procedure ($\mathbb B$) that runs in polynomial time and accepts if and only if $\phi$ is satisfiable.
	The labelling consists of assigning a value $c\in\{0,1\}$ to each variable $x$. 
	\begin{enumerate}
		\item Begin by labelling all $\pl$-literals in $\phi$ by the value that satisfies them, namely $x=1$ for $x$ and $x=0$ for $\neg x$.
		\item For each inclusion atom $\inca{p}{q}$ and a labelled variable $q_i\in \tuple q$, label the variable $p_i\in \mathbf p$ with same value $c$ as for $q_i$. Where $p_i$  appears in $\tuple p$ at the same position, as $q_i$ in $\tuple q$. 
		\item Propagate the label for $p_i$ from the previous step. That is, consider $p_i$ as a labelled variable and repeat Step $2$ for as long as possible.
		\item If 
		some variable $x$ is labelled with two opposite values, then reject. Otherwise, accept.
	\end{enumerate}
		The fact that $\mathbb B$ works in polynomial time is clear because each variable is labelled at most once. 
		If a variable is labelled to two different values, then it gives a contradiction and the procedure stops.
		
		For the correctness, notice first that if $\mathbb B$ accepts then we have a partition of $\VAR(\phi)$ into a set $X$ of labelled variables and a set $Y=\VAR(\phi)\backslash X$.
		When $\mathbb B$ stops, due to step 3, $\phi$ does not contain an inclusion atom $\inca{p}{q}$ such that $q_i\in\mathbf q $ and $p_i\in \mathbf p$ for some $q_i\in X,p_i\in Y$, 
		where $p_i$  appears in $\tuple p$ at the same position as $q_i$ in $\tuple q$. 
		Let $T= \{s \in 2^{\VAR(\phi)} \mid \text{$x$ is labelled with $s(x)$, for each $x\in X$} \}$.
		Since $\mathbb B$ accepts, each variable $x\in X$ has exactly one label and therefore assignments in $T$ are well-defined. 
		Moreover $T$ includes all possible assignments over $Y$.
		One can easily observe that $T\models \phi$.
		$T$ satisfies each literal because each $s\in T$ satisfies it.
		Let $\inca{p}{q}$ be an inclusion atom and $s\in T$ be an assignment.
		We know that for each $x\in \mathbf q$ that is fixed by $s$, the corresponding variable $y\in \mathbf p$ is also fixed,
		whereas, $T$ contains every possible value for variables in $\mathbf q$ which are not fixed.
		This makes the inclusion atom true.
		
		To prove the other direction, suppose that $\mathbb B$ rejects.
		Then there are three cases under which a variable contains contradictory labels.
		Either both labels of the variable are caused by a literal (Case 1), or inclusion atoms are involved in one (Case 2), or both (Case 3) labels.
		In other words, either $\phi$ contains $x\land \neg x$ as a subformula, or it contains $x\land \neg y$ and there is a sequence of inclusion atoms, such that keeping $x=1$ and $y=0$ contradicts some inclusion atoms in $\phi$ (see Figure~\ref{fig:atoms}).
		\begin{description}
			\item[Case 1] Both labels of a variable $x$ are caused by a literal. 
			In this case, $x$ takes two labels because $\phi$ contains $x\land \neg x$. The proof is trivial since $\phi$ is unsatisfiable.
			
			\item[Case 2] One label of a variable $y$ is caused by a literal ($\neg y$ or $y$) and the other by inclusion atoms. 
			Then, there are inclusion atoms $\incas{\mathbf p_j}{\mathbf q_j}$ and variables $z_j$ for $j\leq n$ such that: $z_0=x$, $z_n=y$, and  $z_j$ and $z_{j+1}$ occur in the same position in $\mathbf q_{j}$ and  $\mathbf p_{j}$, respectively, for $0\leq j< n$. 
			This implies that $\phi$ is not satisfiable since for any team $T$ such that $T\models x\land \neg y$, $T$ does not satisfy the subformula  $\bigwedge_j \incas{\mathbf p_j}{\mathbf q_j}$ of $\phi$.
			A similar reasoning applies if $\phi$ contains $\neg x\land y$ instead.
			\item[Case 3] Both labels of a variable $v$ are caused by inclusion atoms.
			Then, there are two collections of inclusion atoms $\incas{\mathbf p_j}{\mathbf q_j}$ for $j\leq n$, and $\incas{\mathbf r_k}{\mathbf s_k}$ for $k\leq m$.
			Moreover, there are two sequences of variables $z^x_j$ for $j\leq n$ and $z^y_k$ for $k\leq m$, and a variable $v$ such that, $z^x_0=x$, $z^y_0=y$, $z^x_n=v=z^y_m$, and
			\begin{enumerate}
				\item for each $j\leq n$, $z^x_j$ appears in $\tuple q_j$ at the same position, as $z^x_{j+1}$ in $\tuple p_j$,
				\item for each $k\leq m$, $z^y_k$ appears in $\tuple s_k$ at the same position, as $z^y_{k+1}$ in $\tuple r_k$.
			\end{enumerate}
			This again implies that $\phi$ is not  satisfiable since for any $T$ such that $T\models x\land \neg y$, it does not satisfy the subformula $\bigwedge_j \incas{\mathbf p_j}{\mathbf q_j}\land\bigwedge_k \incas{\mathbf r_k}{\mathbf s_k}$  of $\phi$.
		\end{description}
		Consequently, the correctness follows. 
	%
	\begin{figure}[t]
		\centering
		\begin{tikzpicture}[auto, sibling distance= 2cm, level distance = 1.5 cm, scale=.7, node1/.style = {scale=0.6}, edge/.style = {} ] 
			
			\node at (0.5,4) (x0) {$x$};
			\node at (1.5,4) (x1) {$z_1$};
			\node at (3,4) (x2) {$z_2$};
			\node at (4.5,4) (x3) {$\ldots$};
			\node at (6,4) (x4) {$z_{n-1}$};
			\node at (7.3,4) (x5) {$y$};

			\node at (9.2,4.5)  (z0) {$x$};
			\node at (10.2,4.5)(z1) {$z^x_1$};
			\node at (11.4,4.5)(z2) {$z^x_2$};
			\node at (12.6,4.5)(z3) {$\ldots$};
			\node at (14,4.5)  (z4) {$z^x_{n-1}$};
			\node at (15.3,4.5)(z5) {$v$};

			\node at (9.2,3.5)   (y0) {$ y$};
			\node at (10.2,3.5)(y1) {$z^y_1$};
			\node at (11.6,3.5)(y2) {$z^y_{2}$};
			\node at (12.8,3.5)(y3) {$\ldots$};
			\node at (14,3.5)  (y4) {$z^y_{m-1}$};
			\node at (15.3,3.5)(y5) {$v$};

			\draw (x0) to node {} (x1) ; 
			\draw (x1) to node {} (x2) ; 
			\draw (x2) to node {} (x3) ; 
			\draw (x3) to node {} (x4) ; 
			\draw (x4) to node {} (x5) ; 
			
			\draw (y0) to node {} (y1) ; 
			\draw (y1) to node {} (y2) ; 
			\draw (y2) to node {} (y3) ; 
			\draw (y3) to node {} (y4) ; 
			\draw (y4) to node {} (y5) ;

			\draw (z0) to node {} (z1) ; 
			\draw (z1) to node {} (z2) ; 
			\draw (z2) to node {} (z3) ; 
			\draw (z3) to node {} (z4) ; 
			\draw (z4) to node {} (z5) ; 
		\end{tikzpicture}
		\caption{Intuitive explanation of two cases in the proof. (Left)  $x$ and $\neg y$ propagate a conflicting value to eachother. (Right) $x$ and $\neg y$ propagate conflicting values to $v$.}\label{fig:atoms}
	\end{figure}
	This completes the proof.
\end{proof}

\begin{example}\label{ex:no-split}
	We include an example to highlights how $\mathbb B$ operates.
	Let $\phi\dfn (x_1\land x_2 \land \neg x_4)\land (x_1x_3\subseteq x_5x_2)\land (x_5\subseteq x_4)$.
	The table in Figure~\ref{fig:algo-no-split} illustrates the steps of $\mathbb B$ on $\phi$.
	Clearly, $\mathbb B$ rejects $\phi$ since the variable $x_1$ has conflicting labels. 
	\begin{figure}[t]
		\centering
		\begin{tikzpicture}[every node/.style={scale=1}]
			\node  {\begin{tabular}{l@{\; }c@{\; }c@{\; }c@{\; }c@{\; }c@{\; }c@{\; }c}\toprule
					$\VAR{(\phi)}$ 	& $x_1$	& $ x_2$ & $x_3 $ & 	$ x_4$ & $x_5$ \\ \toprule
					Labels for $\VAR{(\phi)}$ 		
					&	1 & 1		&  	&  0		&  \\
					
					Propagation	due to $\incas{x_5}{x_4}$	
					& 1 & 1	&  	& 0	& 0  \\
					Propagation	due to $\incas{x_1x_3}{x_5x_2}$	
					& 1/0 & 1	& 1 	& 0	& 0 \\

					\bottomrule
				\end{tabular}				
			};
			
		\end{tikzpicture} 
		\caption{Labels for literals and their propagation to inclusion atoms (See Example~\ref{ex:no-split}).}\label{fig:algo-no-split}
	\end{figure}
\end{example}

\noindent The $\FPT$ cases for $\SAT$ of $\pinc$ follow from Theorem~\ref{SAT:general}.
Regarding $\MC$, recall that we consider strict semantics alone. 
The results of Theorem~\ref{thm:mc-many} are obtained from the reduction for proving $\NP$-hardness of $\MC_s$ for $\pinc$ \cite{HellaKMV19}.
Here we confirm that their reduction is indeed an \emph{fpt}-reduction with respect to considered parameters.
The following lemma is essential for proving Theorem~\ref{thm:mc-many} and we include it for self containment. 
%
\begin{lemma}[{\cite
		{HellaKMV19}}]\label{lem:NPhard}
	$\MC$ for $\pinc$ under strict semantics is $\NP$-hard.
\end{lemma}\label{mc-lem-pinc}
%
\begin{proof}[Proof Idea]
	The hardness is achieved through a reduction from the set splitting problem to the model checking problem for $\pinc$ with strict semantics.
	An instance of set splitting problem consists of a family $\mathcal F$ of subsets of a finite set $S$.
	The problem asks if there are $S_1,S_2\subseteq S$ such that $S_1\cup S_2 =S$, $S_1\cap S_2 =\emptyset$ and for each $A\in \mathcal F$ there exists $a_1,a_2\in \calA$ such that $a_1\in S_1,a_2\in S_2$.
	Let $\mathcal F= \{B_1,\ldots, B_n\}$ and $\bigcup \mathcal F =S = \{a_1,\ldots, a_k\}$.
	Let $p_i$ and $q_j$ denote fresh variables for each $a_i\in  S$ and $B_j\in \mathcal F$.
	Moreover, let $ V_{\mathcal F}=\{p_1,\ldots, p_k, q_1,\ldots, q_n, p_{\top},p_c,p_d \}$.
	Then define $T_{\mathcal F}=\{s_1,\ldots, s_k,s_c,s_d\}$, where each assignment $s_i$ is defined as follows:
	\[
	s_i(p) \dfn \begin{cases}
		1 , \text{ if } p=p_i \text{ or } p=p_\top,&\\
		1, \text{ if } p=q_j \text{ and } a_i\in B_j \text{ for some } j, &\\
		0, \text{ otherwise.}
	\end{cases}\label{abcd}
	\]
	That is, $T_\mathcal F$ includes an assignment $s_i$ for each  $a_i\in S$.
	The reduction also yields the following $\pinc$-formula.
	\[
	\phi_\mathcal{F} \dfn (\neg p_c \land \bigwedge\limits_{i\leq n}p_\top \subseteq q_i)\lor  (\neg p_d \land \bigwedge\limits_{i\leq n}p_\top \subseteq q_i)
	\]
	Clearly, the split of $T_\mathcal F$ into $T_1,T_2$  ensures the split of $S$ into $S_1$ and $S_2$ and vice versa.
	Whereas, $s_c$ and $s_d$ ensure that none of the split is empty.  
\end{proof}

\begin{theorem}\label{thm:mc-many}
	$\MC_s$ for $\pinc$ when parameterized by $k$ is $\para\NP$-complete if $k\in \{\splits, \arity, \formulatw \}$. 
	Whereas, it is $\FPT$ in other cases.
\end{theorem}

\begin{proof}	
	Consider the $\pinc$-formula $\phi_\mathcal{F}$ from Lemma~\ref{lem:NPhard}, which includes only one split-junction and the inclusion atoms have arity one. 
	This gives the desired $\para\NP$-hardness for $\MC_s$ when parameterized by $\splits$ and $\arity$.

	The proof for $\formulatw$ is more involved and we prove the following claim.
	\begin{claim}\label{claim:MC-tw}
		$\phi_\mathcal{F}$ has fixed $\formulatw$. That is, the treewidth of $\phi_\mathcal{F}$ is independent of the input instance $\mathcal{F}$ of the set-spliting problem. 
		Moreover $\formulatw(\phi_\mathcal F)\leq 4$. 
	\end{claim}
	\begin{claimproof} 
		The $\pinc$-formula $\phi_{\mathcal F}$ is related to an input instance ${\mathcal F}$ of the set splitting problem only through its input size, which is $n$. 
		Therefore the formula structure remains unchanged when we vary an input instance, only the size of two big conjunctions vary. 
		To prove the claim, we give a tree decomposition for the formula with $\formulatw(\phi_\mathcal{F})=4$.
		Since the treewidth is minimum over all tree decompositions, this proves the claim.
		We rewrite the formula as below.
		\[
		\phi_\mathcal{F} \dfn (\neg p_c \land_l \bigwedge\limits_{i\leq n}p_\top \subseteq^l_i q_i)\lor  (\neg p_d \land_r \bigwedge\limits_{i\leq n} p_\top \subseteq^r_i q_i)
		\]
		That is, each subformula is renamed so that it is easy to identify as to which side of the split it appears (e.g., $p_\top\subseteq^l_i q_i$ denotes the $i$th inclusion atom in the big conjunction on the left,  denoted as $I^l_i$ in the graph).
		The graphical representation of $\phi_{\mathcal F}$ with $V= \SubForm{\phi_{\mathcal F}}\cup \VAR(\phi_{\mathcal F})$, as well as, a tree decomposition, is given in Figure~\ref{fig:tw}.
		Notice that there is an edge between $x$ and $y$ in the Gaifman graph if and only if either $y$ is an immediate subformula of $x$, or $y$ is a variable appearing in the inclusion atom $x$.
		It is easy to observe that the decomposition presented in Figure~\ref{fig:tw} is indeed a valid tree decomposition in which each node is labelled with its corresponding bag.
		Moreover, since the maximum bag size is $5$, the treewidth of this decomposition is $4$. 
		This proves the claim. \qedclaim
		\begin{figure}[t]
			\centering
			\begin{tikzpicture}[rotate=90,scale=.5, every node/.style={scale=0.7}]
				\node (qn) at (-0.2,0) {$q_n$};
				\node (d1) at (-0.2,0.9) {\small$\ldots$};
				\node (q2) at (-0.2,1.6) {$q_2$};
				\node (q1) at (-0.2,2.6) {$q_1$};
				
				\node (pt) at (2.5,1.5) {$p_\top$};
				
				\node (in) at (1.2,1.9) {\small$I^l_n$};
				\node (d2) at (1.2,2.5) {\scriptsize$\tiny\ldots$};
				\node (i2) at (1.2,3.2) {\small$I^l_2$};
				\node (i1) at (1.2,4) {\small$I^l_1$};
				
				\node (rin) at (1.2,-1.1) {\small$I^r_n$};
				\node (rd2) at (1.2,-0.5) {\scriptsize$\ldots$};
				\node (ri2) at (1.2,0.1) {\small$I^r_2$};
				\node (ri1) at (1.2,0.9) {\small$I^r_1$};
				
				\node (al) at (4,3) {$\bigwedge_L$};
				\node (ar) at (4,0) {$\bigwedge_R$};
				\node (ngr) at (5.3,-1) {$\neg_r$};
				\node (ngl) at (5.3,4) {$\neg_l$};
				
				\node (pc) at (4,-1) {$p_d$};
				\node (pd) at (4,4) {$p_c$};
				
				\node (wl) at (5.3,3) {$\wedge_l$};
				\node (wr) at (5.3,0) {$\wedge_r$};
				\node (l) at (6.6,1.5) {$\vee$};

				\foreach \f/\t in {pt/i1, pt/i2,pt/in,pt/ri1, pt/ri2,pt/rin,q1/i1,q2/i2,qn/in, q1/ri1,q2/ri2,qn/rin, pc/ngr,pd/ngl, al/i1, al/i2, al/in, ar/rin, ar/ri1, ar/ri2, ar/rin, ar/wr, al/wl, ngl/wl, ngr/wr, wr/l,wl/l}{
					\draw[-] (\f) -- (\t);
				}
				
			\end{tikzpicture}
			\qquad
			\begin{tikzpicture}[scale=0.55, rounded corners, rectangle, 
				level 1/.style={sibling distance=9em, level distance= 3 em},
				level 2/.style={sibling distance=8em, level distance= 4 em}, 
				level 4/.style={sibling distance=5.7em, level distance= 4 em},
				every node/.style={draw, scale=0.6},
				edge from parent/.style={thin,-,black, draw}]
				\node  {\textcolor{black}{$\wedge_l,\wedge_r,\lor$}}
				child {node {$\wedge_l,\wedge_r$} 
					child {node {$\wedge_l,\neg_l$}
						child {node {$\neg_l,p_c$}}}		
					child {node {$\wedge_l, \bigwedge_L, \wedge_r, \bigwedge_R$}	
						child{node {$\bigwedge_L, \bigwedge_R, p_\top$}
							child{node {\small$\bigwedge_L, \bigwedge_R, p_\top, I^l_1,I^r_1$}
								child{node {\small$p_\top, I^l_1,I^r_1,q_1$}}}
							child[] {node[draw=white] {\small$\ldots$}
								child[]{node[draw=white] {\small$\ldots$}}}
							child {node {\small$\bigwedge_L, \bigwedge_R, p_\top, I^l_n,I^r_n$}
								child{node {\small$p_\top,I^l_n,I^r_n,q_n$}}}}}
					child {node {$\wedge_r,\neg_r$}
						child {node {$\neg_r,p_d$}}}};
			\end{tikzpicture}
			
			\caption{The Gaifman graph (Left) and a tree decomposition (Right) for $\phi_{\mathcal F}$. Note that we abbreviated subformulas in the inner vertices of the Gaifman graph for presentation reasons. Also, edges between $p_\top$ and variables $q_i$ are omitted for better presentation, but those are covered in the decomposition on the right.}\label{fig:tw}
		\end{figure}
	\end{claimproof}
	The remaining $\FPT$-cases for $\MC_s$ follow from Theorem~\ref{MC:general} and Corollary~\ref{MC:cor}.
	This completes the proof to our theorem.
\end{proof} 

Recall that a dependence atom $\depa{x}{y}$ is equivalent with the independence atom $\indepa{y}{y}{x}$.
As a consequence, (in the classical setting) hardness results for $\pdl$ immediately translate to those for $\pind$.
Nevertheless, in the parameterized setting, one has to further check whether this translation `respects' the parameter value of the two instances.
This concerns the parameter $\arity$ and $\formulatw$. 
This is due to the reason that, a dependence atom $\depa{x}{y}$ has arity $|\mathbf{x}|$, whereas, the equivalent independence atom $\indepa{y}{y}{x}$ has arity $|\mathbf{x}\cup \mathbf{y}|$.

\begin{theorem}\label{thm:ind-mc}
$\MC$ for $\pind$, when parameterized by $k$ is $\para\NP$-complete if $k\in \{\splits,\arity,\formulatw \}$. Whereas, it is $\FPT$ in other cases.
\end{theorem}

\begin{proof}
Notice that $\MC$ for $\pdl$ when parameterized by $k\in \{\arity,\splits, \linebreak\formulatw \}$ is also $\para\NP$-complete.
We argue that in reductions for $\pdl$, replacing dependence atoms by the equivalent independence atoms yield \emph{fpt}-reduction for the above mentioned cases.
Moreover, this holds for both strict and lax semantics.


For $\formulatw$ and $\arity$, when proving $\para\NP$-hardness of $\pdl$, the resulting formula has treewidth of one \cite[Cor.~16]{jpdl} and the arity is zero \cite[Thm.~15]{jpdl}. 
Moreover, only dependence atoms of the form $\depas{}{p}$ where $p$ is a propositional variable, are used and the syntax structure of the $\pdl$-formula is already a tree.
Consequently, replacing $\depas{}{p}$ with $\indepas{p}{p}{\emptyset}$ implies that only independence atoms of arity $1$ are used. 
Notice also that replacing dependence atoms by independence atoms does not increase the treewidth of the input formula. 
This is because when translating dependence atoms into independence atoms, no new variables are introduced. 
As a result, the reduction also preserves the treewidth. 
This proves the claim as $1$-slice regarding both parameters $\arity$ and $\formulatw$, is $\NP$-hard.

Regarding the $\splits$, the claim follows due to Mahmood and Meier~\cite[Thm.~18]{jpdl} because the reduction from the colouring problem uses only 2 splits.

Finally, the $\FPT$ cases follow from Theorem~\ref{MC:general} and Corollary~\ref{MC:cor}.
\end{proof}

\begin{theorem}\label{thm:ind-sat}
$\SAT$ for $\pind$,  parameterized by $\arity$ is $\para\NP$-complete.
Whereas, it is $\FPT$ in other cases.
\end{theorem}
\begin{proof}
Recall that $\pl$ is a fragment of $\pind$.
This immediately gives $\para\NP$-hardness when parameterized by $\arity$, because $\SAT$ for $\pl$ is $\NP$-complete.
The $\para\NP$-membership is clear since $\SAT$ for $\pind$ is also $\NP$-complete~\cite[Thm.~1]{Hannula15c}. 
The $\FPT$ cases for $k\in\{\formuladepth,\variables\}$ follow because of Theorem~\ref{SAT:general}\,. 
The cases for $\splits$ and $\formulatw$ follow due to a similar reasoning as in $\pdl$ \cite{jpdl} because it is enough to find a singleton satisfying team \cite[Lemma 4.2]{HannulaKVV18}.
This completes the proof.
\end{proof}

\section{Concluding Remarks}
We presented a parameterized complexity analysis for $\pinc$ and $\pind$.
The problems we considered were satisfiability and model checking.
Interestingly, the parameterized complexity results for $\pind$ coincide with that of $\pdl$~\cite{jpdl} in each case.
Moreover, the complexity of model checking under a given parameter remains the same for all three logics. 
We proved that for a team based logic $\mathcal L$ such that $\mathcal L$-atoms can be evaluated in $\Ptime$-time, $\MC$ for $\mathcal L$ when parameterized by $\teamsize$ is always $\FPT$.

It is interesting to notice that for $\pdl$ and $\pind$, $\SAT$ is easier than $\MC$ when parameterized by $\formulatw$.
This is best explained by the fact that $\pdl$ is downwards closed and a formula is satisfiable iff some singleton team satisfies it.
Moreover, $\pind$ also satisfies this `satisfiable under singleton team' property.
%
The parameters $\teamsize$ and $\formulateamtw$ are not meaningful for $\SAT$ due to the reason that we do not impose a size restriction for the satisfying team in $\SAT$. 
Furthermore, $\arity$ is quite interesting because $\SAT$ for all three logics is $\para\NP$-complete.
This implies that while the fixed $\arity$ does not lower the complexity of $\SAT$ in $\pdl$ and $\pind$, it does lower it from $\EXP$-completeness to $\NP$-completeness for $\pinc$.
As a byproduct, we obtain that the complexity of satisfiability for the fixed arity fragment of $\pinc$ is $\NP$-complete.
Thereby, we answer an open question posed by Hella and Stumpf~\cite[P.13]{HellaStumpf15}.
The $\para\NP$-membership of $\SAT$ when parameterized by $\arity$ implies that one can encode the problem into classical satisfiability and employ a SAT-solver to solve satisfiability for the fixed arity fragment of $\pinc$. 
We leave as a future work the suitable SAT-encoding for $\pinc$ that runs in $\FPT$-time and enables one to use SAT-solvers.
Further future work involves finding the precise complexity of $\SAT$ for $\pinc$ when parameterized by $\splits$ and $\formulatw$.


\bibliography{main.bib}

\end{document}